\documentclass[a4paper, 10pt, twosides]{amsart}

\usepackage{lmodern}
\usepackage{amssymb}
\usepackage{amsthm}
\usepackage{amsaddr}
\usepackage{mathtools}
\usepackage{MnSymbol}
\usepackage{graphicx}
\usepackage{caption}
\usepackage{subcaption}
\usepackage{url}

\usepackage{enumitem}
\setitemize{nolistsep}
\setenumerate{nolistsep}
\usepackage{todonotes}
\usepackage{algorithm,algorithmic}

\allowdisplaybreaks
\newtheorem{thm}{Theorem}
\newtheorem{defn}{Definition}
\newtheorem{fact}{Fact}
\newtheorem{assump}{Assumption}

\newtheorem{remark}{Remark}
\newtheorem{example}{Example}
\newtheorem{prob}{Problem}

\newtheorem{lemma}{Lemma}

\newcommand{\pmat}[1]{\begin{pmatrix}#1\end{pmatrix}}

\renewcommand{\geq}{\geqslant}

\renewcommand{\leq}{\leqslant}

\newcommand{\abs}[1]{\left\lvert{#1}\right\rvert}

\newcommand{\norm}[1]{\left\lVert#1\right\rVert}

\newcommand{\R}{\mathbb{R}}
\newcommand{\N}{\mathbb{N}}
\newcommand{\is}{i_{s}}
\newcommand{\iu}{i_{u}}
\newcommand{\Sc}{\mathcal{S}}
\newcommand{\V}{\mathcal{T}}
\newcommand{\A}{\mathcal{A}}

\title[]{Scheduling networked control systems\\under jamming attacks}
\author{Atreyee Kundu}
\address{Department of Electrical Engineering,\\Indian Institute of Science Bangalore,\\Bengaluru - 560012, India,\\ E-mail: atreyeek@iisc.ac.in}
\thanks{The author thanks Daniel Quevedo for helpful discussions and comments on an earlier version of this manuscript.}


\date{\today}

\begin{document}
    \begin{abstract}
	This paper deals with the design of scheduling policies for networked control systems whose shared networks have limited communication capacity and the controller to plant channels are vulnerable to jamming attacks. We assume that among \(N\) plants, only \(M (< N)\) plants can communicate with their controllers at any time instant, and the attack sequences follow an \((m,k)\)-firm model, i.e., in any \(k\) consecutive time instants, the control inputs sent to some or all of the plants accessing the communication network, are deactivated at most at \(m (< k)\) time instants. We devise a new algorithm to allocate the network to the plants periodically such that stability of each plant is preserved under the admissible attack signals. The main apparatus for our analysis is a switched systems representation of the individual plants in an NCS. We rely on matrix commutators (Lie brackets) between the stable and unstable modes of operation of the plants to guarantee stability under our scheduling policies.
\end{abstract}

\maketitle

\section{Introduction}
\label{s:intro}
    Networked Control Systems (NCSs) are spatially distributed control systems in which the communication between plants and their controllers occurs through shared networks. Modern day Cyber-Physical Systems (CPS) and Internet-of-Things (IoT) applications of NCSs typically involve a large number of plants. However, bandwidth of the shared communication network is often limited. The scenario in which the number of plants sharing a communication network is higher than the capacity of the network is called \emph{medium access constraint}. This scenario motivates the need to allocate the communication network to each plant in a manner so that good qualitative properties of the plants are preserved. This task of efficient allocation of a shared communication network is commonly referred to as a \emph{scheduling problem}, and the corresponding allocation scheme is called a \emph{scheduling policy}.

    The problem of designing scheduling policies for NCSs has been addressed widely in the past, see e.g., \cite{Hristu2008,Rehbinder2004,Al-Areqi2015,Quevedo2014,Demirel2018,Miao2017,abc,Peters2016,Ma2019} and the references therein. Under ideal communication, both a plant and its controller receive the intended information whenever the shared network is allocated to them. However, in practical situations, communication networks are often vulnerable to malicious attacks. The study of qualitative properties of NCSs under attacks on the shared network have recently attracted considerable research attention, see e.g., \cite{Srikant2016,Ding2019,Heemels2017,Guo2019} and the references therein. In this paper we consider jamming attacks, where the attacker (jammer) deactivates the control inputs sent to the plants intermittently. This aspect further leads to the requirement of designing scheduling policies that preserve good qualitative properties of the plants resilient to jamming attacks. In this paper we address this requirement.

    We consider an NCS consisting of multiple discrete-time linear plants whose feedback loops are closed through a shared communication network. 
    We assume that the plants are unstable in open-loop and exponentially stable when controlled in closed-loop. Due to a limited communication capacity of the network, only some of the plants (\(M\) out of \(N\)) can exchange information with their controllers at any instant of time. Consequently, the remaining plants operate in open-loop potentially leading to instability. In addition, the channels from the controllers to the plants (more specifically, the controllers to the actuators channels) are vulnerable to jamming attacks. The jammer deactivates the control inputs sent to the plants accessing the shared network. She does not want her presence to be detected by the plants and has certain limitations on the operations of the jamming devices. As a consequence, she chooses to deactivate the control inputs sent to some or all of the plants accessing the shared network only intermittently. In particular, in every \(k\) consecutive time instants, she deactivates control inputs at most at \(m (< k)\) time instants (and hence the flow of control inputs remains unaffected at least at \(k-m\) time instants).\footnote{The jamming attack model considered in this paper is a discrete-time counterpart of the `persistence of excitation' attack model employed in \cite{Srikant2016}. Moreover, the \((m,k)\)-firm model is used widely to characterize data loss signals in NCSs, see e.g., \cite{Jungers2018,Heemels2016}.} Our objective is to design scheduling policies that preserve global exponential stability (GES) of each plant in the NCS under the jamming attacks described above.

    We model the individual plants of an NCS as switched systems, where the switching is between their open-loop (unstable mode) and closed-loop (stable mode) operations. The open-loop operation of a plant occurs at the time instants when it does not have access to the shared network, or it has access to the shared network but its control input is deactivated by the jammer. The closed-loop operation occurs at the time instants when the plant has access to the shared network and its control input is not deactivated by the jammer. Clearly, a switching logic is governed by both the scheduling policy and the attack signal. We design periodic scheduling policies in a manner such that switching logic corresponding to each plant is stabilizing under all admissible attack signals. Towards this end, we ensure a certain number of closed-loop operation instances of a set of plants prior to allocating the network to another set of plants. This number is related to a measure of the rate of decay of the stable dynamics of each plant. We provide conditions on the matrix commutators (Lie brackets) between the open-loop and the closed-loop dynamics of the plants such that stability of all plants in the NCS is preserved under our scheduling policies. The proposed scheduling policies are \emph{static} \cite{abc} in the sense that an allocation sequence of the shared network is computed offline and it is repeated eternally obeying a certain logic.
%

    Recently in \cite{abc} we employed a switched systems modelling of individual plants in an NCS to design stabilizing scheduling policies under ideal communication. We relied on multiple Lyapunov-like functions and graph-theoretic tools for this purpose. In this paper we consider jamming attacks and introduce matrix commutator based tools for the design of stabilizing scheduling policies for NCSs. Our analysis for stability of the individual plants relies on the application of combinatorial arguments to matrix products split into sums. This extends the techniques of \cite{Agrachev2012} to the analysis of simultaneous stability of \(N\) switched systems where the unstable dynamics of each system obeys a certain maximum dwell time. To the best of our knowledge, this is the first instance in the literature where matrix commutation relations between the stable and unstable modes of operations of the individual plants are employed to guarantee stability under scheduling. 

    The remainder of this paper is organized as follows: in \S\ref{s:prob_stat} we formulate the problem under consideration. The tools for our design of scheduling policies and analysis of stability are described in \S\ref{s:prelims}. Our results appear in \S\ref{s:mainres}. We also discuss various features of our results in this section. Numerical experiments are presented in \S\ref{s:numex}. We conclude in \S\ref{s:concln}.

    {\bf Notation}. \(\R\) is the set of real numbers and \(\N\) is the set of natural numbers, \(\N_{0} = \N\cup\{0\}\). \(\norm{\cdot}\) denotes the Euclidean norm (resp., induced matrix norm) of a vector (resp., a matrix). For a finite set \(A\), we employ \(\abs{A}\) to denote its cardinality, i.e., the number of elements in \(A\). For a matrix \(P\), given by a product of matrices \(C_{j}\)'s, \(\abs{P}\) denotes the length of the product, i.e., the number of matrices that appear in \(P\), counting repetitions. For \(a\in\R\), \(\lfloor a\rfloor\) denotes the largest integer less than or equal to \(a\).

\section{Problem statement}
\label{s:prob_stat}
    We consider an NCS with \(N\) plants whose dynamics are given by
    \begin{align}
    \label{e:plants}
        x_{i}(t+1) = A_{i}x_{i}(t) + B_{i}u_{i}(t),\:\:x_{i}(0) = x_{i}^{0},\:\:t\in\N_{0},
    \end{align}
    where \(x_{i}(t)\in\R^{d}\) and \(u_{i}(t)\in\R^{m}\) are the vectors of states and control inputs of the \(i\)-th plant at time \(t\), respectively, \(i=1,2,\ldots,N\). Each plant \(i\) employs a state-feedback controller \(u_{i}(t) = K_{i}x_{i}(t)\), \(t\in\N_{0}\). The matrices \(A_{i}\in\R^{d\times d}\), \(B_{i}\in\R^{d\times m}\) and \(K_{i}\in\R^{m\times d}\), \(i=1,2,\ldots,N\) are constant.
    \begin{assump}
    \label{a:stab_unstab}
    \rm{
        The open-loop dynamics of each plant is unstable and each controller is stabilizing. More specifically, the matrices \(A_{i}+B_{i}K_{i}\), \(i=1,2,\ldots,N\) are Schur stable and the matrices \(A_{i}\), \(i=1,2,\ldots,N\) are unstable.\footnote{A matrix \(A\in\R^{d\times d}\) is Schur stable if all its eigenvalues are inside the open unit disk. We call \(A\) unstable if it is not Schur stable.}
    }
    \end{assump}

    The controllers are remotely located, and each plant communicates with its controller through a shared communication network. We consider that the shared network has the following properties:\\
    (i) It has a limited communication capacity in the sense that at any time instant, only \(M\) plants \((0 < M < N)\) can access the network. Consequently, the remaining \(N-M\) plants operate in open-loop.\\
    (ii) The communication channels from the controllers to the plants are prone to jamming attacks. The jammer corrupts the control input \(u_{i}(t)\) for some or all \(i\in\{1,2,\ldots,N\}\) such that \(i\) is accessing the shared communication network at time \(t\), with a jamming input \(v_{i}(t)\in\{0,1\}\) that enters multiplicatively. In particular, the effect of \(v_{i}(t) = 0\) is that the control input \(u_{i}(t)\) is deactivated (i.e., set to \(0\)) and the \(i\)-th plant operates in open-loop at time \(t\) (even though it has access to the shared network). The jammer follows an \((m,k)\)-firm model, i.e., in any \(k\) consecutive time instants, the attacker selects \(v_{i}(t) = 0\) for some or all \(i\) accessing the shared network at most at \(m\) consecutive time instants, \(m < k\).

    In view of Assumption \ref{a:stab_unstab} and the properties of the shared communication network, each plant in \eqref{e:plants} operates in two modes: (a) stable mode (closed-loop operation) when the plant has access to the shared communication network and its control input is not deactivated by the jammer, and (b) unstable mode (open-loop operation) when the plant does not have access to the shared communication network, or when it has access to the shared network but its control input is deactivated by the jammer.

    Let \(\is\) and \(\iu\) denote the stable and unstable modes of the \(i\)-th plant, respectively, \(A_{\is} = A_{i}+B_{i}K_{i}\) and \(A_{\iu} = A_{i}\), \(i=1,2,\ldots,N\). We let \(\Sc\)
    be the set of all vectors that consist of \(M\) distinct elements from the set \(\{1,2,\ldots,N\}\). We call a function \(\gamma:\N_{0}\to\Sc\) that specifies, at every time \(t\), \(M\) plants of the NCS which has access to the shared network at that time, as a \emph{scheduling policy}. There exists a diverging sequence of times \(0=:\tau_{0}<\tau_{1}<\tau_{2}<\cdots\) and a sequence of indices \(s_{0}\), \(s_{1}\), \(s_{2},\ldots\) with \(s_{j}\in\Sc\), \(j=0,1,2,\ldots\) such that \(\gamma(t) = s_{j}\) for \(t\in[\tau_{j},\tau_{j+1}[\), \(j=0,1,2,\ldots\). Let \(\V\) be the set of all \(N\)-dimensional vectors whose at most \(M\) elements are \(0\) and the remaining elements are \(1\).\footnote{Notice that \(\V\) contains the \(N\)-dimensional vector whose all elements are \(1\).} We let \(\kappa:\N_{0}\to\V\) denote a jamming attack signal, defined as \(\kappa(t) = \pmat{v_{1}(t) & v_{2}(t) & \cdots & v_{N}(t)}^\top\), where \(v_{\ell}(t) = 0\) for at most \(M\) elements \(\ell\in\{1,2,\ldots,N\}\) such that \(\ell\) is an element of \(\gamma(t)\), and \(v_{\ell}(t) = 1\) for the remaining elements of \(\kappa(t)\). We call \(\kappa\) admissible if it obeys the \((m,k)\)-firm model, \(m,k\) given. Let \(\A(m,k)\) denote the set of all admissible \(\kappa\). Fix \(t\in\N_{0}\). Let \(\gamma_{\ell}(t)\) denote the \(\ell\)-th element of \(\gamma\) at time \(t\), \(\ell = 1,2,\ldots,M\). A plant \(i\in\{\gamma_{1}(t),\gamma_{2}(t),\ldots,\gamma_{M}(t)\}\) operates in closed-loop at time \(t\), if \(v_{i}(t) = 1\), and operates in open-loop at time \(t\), if \(v_{i}(t) = 0\). All plants \(j\in\{1,2,\ldots,N\}\setminus\{\gamma_{1}(t),\gamma_{2}(t),\ldots,\gamma_{M}(t)\}\) operate in open-loop at time \(t\).

    We will solve the following problem:
    \begin{prob}
    \label{prob:mainprob}
    \rm{
        Given the matrices \(A_{i}\), \(B_{i}\), \(K_{i}\), \(i=1,2,\ldots,N\) and the numbers \(M\), \(m\) and \(k\), design scheduling policies \(\gamma\) that preserve stability of each plant \(i\) in \eqref{e:plants} under all admissible attack signals \(\kappa\in\A(m,k)\).
    }
    \end{prob}
    In particular, we are interested in global exponential stability (GES) of the plants:
    \begin{defn}{\cite[\S 2]{Agrachev2012}}
    \label{d:ges}
    \rm{
        We call the \(i\)-th plant in \eqref{e:plants} GES for a given scheduling policy \(\gamma\), if there exist positive numbers \(c_{i}\) and \(\lambda_{i}\) such that for arbitrary choice of the initial condition \(x_{i}^{0}\), the following inequality holds:
        \begin{align}
        \label{e:ges}
            \norm{x_{i}(t)}\leq c_{i}e^{-\lambda_{i}t}\norm{x_{i}^{0}}\:\text{for all}\:t\in\N_{0}\:\text{and all}\:\kappa\in\A(m,k).
        \end{align}
    }
    \end{defn}

    In the sequel we will call a \(\gamma\) that preserves GES of each plant in \eqref{e:plants} as a \emph{stabilizing scheduling policy}. Prior to presenting our solution to Problem \ref{prob:mainprob}, we catalog a set of preliminaries.
\section{Preliminaries}
\label{s:prelims}
\subsection{Individual plants and switched systems}
\label{ss:plants-swsys}
    Fix \(i\in\{1,2,\ldots,N\}\). The dynamics of plant \(i\) can be expressed as a switched system \cite[\S1.1.2]{Liberzon}
   \begin{align}
   \label{e:i-swsys}
    x_{i}(t+1) = A_{\sigma_{i}(t)}x_{i}(t),\:x_{i}(0) = x_{i}^{0},\:t\in\N_{0},
   \end{align}
   where the subsystems are \(\{A_{\is},A_{\iu}\}\) and a switching logic \(\sigma:\N_{0}\to\{\is,\iu\}\) satisfies
   \begin{align*}
        \sigma_{i}(t) =
        \begin{cases}
            \is,\:&\text{if}\:i\:\text{is an element of \(\gamma(t)\) and \(v_{i}(t) = 1\)},\\
            \iu,\:&\text{if}\:i\:\text{is not an element of \(\gamma(t)\)},\\&\:\:\text{or \(i\) is an element of \(\gamma(t)\) and \(v_{i}(t) = 0\)}.
        \end{cases}
   \end{align*}
\subsection{The stable mode of plant \(i\)}
\label{ss:i-stab_mode}
    The following fact is immediate from the properties of Schur stable matrices:
    \begin{fact}
    \label{fact:key}
    \rm{
        There exist \(\N\ni\delta\geq m\) and \(\rho\in]0,1[\) such that
        \begin{align}
        \label{e:key}
            \norm{A_{\is}^{\delta}}\leq\rho\:\:\text{for all}\:i=1,2,\ldots,N.
        \end{align}
    }
    \end{fact}
    \begin{remark}
    \label{rem:delta}
    \rm{
        Since the matrices \(A_{\is}\), \(i=1,2,\ldots,N\) are Schur stable, it follows that there exist \(\N\ni\delta_{i}\geq m\), \(i=1,2,\ldots,N\), such that \(\norm{A_{\is}^{\hat{\delta}_{i}}}< 1\) for all \(\hat{\delta}_{i}\geq\delta\), \(i=1,2,\ldots,N\). A natural choice of \(\displaystyle{\delta = \max_{i=1,2,\ldots,N}\delta_{i}}\).
    }
    \end{remark}
    Our design of stabilizing scheduling policies will involve \(\delta\)-many closed-loop operation instances of a set of \(M\) plants prior to assigning the network to a different set of plants.\\
\subsection{The attack signals}
\label{ss:attack-seq}
    We provide an estimate of the maximum number of time instants required for at least \(\delta\)-many closed-loop operation instances of the plants accessing the shared network under an admissible jamming attack signal.
    \begin{lemma}
    \label{lem:attack}
    \rm{
        Fix \(i\in\{1,2,\ldots,N\}\). Let \(\delta, m, k\in\N\), \(m < k\), be given, and \(t_{1}\), \(t_{2},\ldots,t_{T}\) be the consecutive time instants such that for each \(p=1,2,\ldots,T\), there exists \(\ell\in\{1,2,\ldots,M\}\) such that \(\gamma_{\ell}(t_{p}) = i\). Suppose that there exist \(\delta\)-many \(t_{p}\) such that \(v_{i}(t_{p}) = 0\), \(p\in\{1,2,\ldots,T\}\). Then
        \(\displaystyle{
            T \leq \bigl(\big\lfloor\frac{\delta}{m}\big\rfloor+1\bigr)(k-m)+\delta.}
        \)
        }
    \end{lemma}
    \begin{proof}
        We rewrite \(\delta\) as
        \(
            \big\lfloor\frac{\delta}{m}\big\rfloor m + \biggl(\delta - \big\lfloor\frac{\delta}{m}\big\rfloor m\biggr).
        \)
        Let \(T = T_{1} + T_{2}\), where \(t_{1},t_{2},\ldots,t_{T_{1}}\) be the time instants that contain \(\big\lfloor\frac{\delta}{m}\big\rfloor m\)-many \(t_{p}\) such that \(v_{i}(t_{p}) = 0\), \(p\in\{1,2,\ldots,T_{1}\}\) and \(t_{T_{1}+1},t_{T_{1}+2},\ldots,t_{T_{1}+T_{2}}\) be the time instants that contain \(\biggl(\delta-\big\lfloor\frac{\delta}{m}\big\rfloor m\biggr)\)-many \(t_{p}\) such that \(v_{i}(t_{p}) = 0\), \(p\in\{T_{1}+1,T_{1}+2,\ldots,T_{1}+T_{2}\}\).

        From the properties of an attack signal, we have
        \[
            T_{1}\leq\big\lfloor\frac{\delta}{m}\big\rfloor k\:\:\text{and}\:\:T_{2}\leq (k-m)+\biggl(\delta-\big\lfloor\frac{\delta}{m}\big\rfloor m\biggr).
        \]
        Consequently,
        \begin{align*}
            T = T_{1} + T_{2} &\leq \big\lfloor\frac{\delta}{m}\big\rfloor k + (k-m) + \delta - \big\lfloor\frac{\delta}{m}\big\rfloor m\\
            & =\big\lfloor\frac{\delta}{m}\big\rfloor (k-m) + (k-m) + \delta\\
            & =\bigl(\big\lfloor\frac{\delta}{m}\big\rfloor+1\bigr)(k-m) + \delta.
        \end{align*}
    \end{proof}
\subsection{Definitions}
\label{ss:defn+term}
    We let the commutators of the matrices \(A_{\is}\) and \(A_{\iu}\), \(i=1,2,\ldots,N\) be given by
    \begin{align}
    \label{e:comm}
        E_{\is\iu} := A_{\is}A_{\iu} - A_{\iu}A_{\is}.
    \end{align}
    Let
    \begin{align}
    \label{e:r_defn} r &:= \begin{cases}
        \lfloor\frac{N}{M}\rfloor,\:\:\text{if}\:N-\lfloor\frac{N}{M}\rfloor M = 0,\\
        \lfloor\frac{N}{M}\rfloor+1,\:\:\text{otherwise},
    \end{cases}
    \end{align}
    \begin{align}
    \label{e:alpha_defn} \alpha(\delta) &:= \biggl(\bigl\lfloor\frac{\delta}{m}\bigr\rfloor+1\biggr)(k-m)+\delta,\\
    \label{e:zeta_defn} \zeta(\delta) &:= r\alpha(\delta) - \delta - 1,\\
    \intertext{and}
    \label{e:eta_defn} \eta(\delta) &:= m\biggl(\frac{\bigl\lfloor\frac{\delta}{m}\bigr\rfloor\bigl(\bigl\lfloor\frac{\delta}{m}\bigr\rfloor+1\bigr)}{2}(k-m)+(r-1)\alpha(\delta)\biggr)\nonumber\\
    &\hspace*{-0.5cm}+\biggl(\delta-\bigl\lfloor\frac{\delta}{m}\bigr\rfloor m\biggr)\biggl(\bigl(\bigl\lfloor\frac{\delta}{m}\bigr\rfloor+1\bigr)(k-m)+(r-1)\alpha(\delta)\biggr).
    \end{align}

    We are now in a position to present our results.
\section{Results}
\label{s:mainres}
    We first present an algorithm that constructs scheduling policies. We will then provide sufficient conditions on the matrices \(A_{\is}\) and \(A_{\iu}\), \(i=1,2,\ldots,N\), such that a scheduling policy obtained from our algorithm preserves GES of all plants in \eqref{e:plants}.
    \begin{algorithm}[htbp]
			\caption{Construction of scheduling policies} \label{algo:scheduling}
		\begin{algorithmic}[1]
			\renewcommand{\algorithmicrequire}{\textbf{Input:}}
			\renewcommand{\algorithmicensure}{\textbf{Output:}}
			
			\REQUIRE The numbers \(M\), \(N\), \(k\), \(m\) and \(\delta\)
			\ENSURE A scheduling policy \(\gamma\)
			
            \hspace*{-0.6cm}Step I: Create distinct vectors \(v\in\Sc\) until each element of \(\{1,2,\ldots,N\}\) appears once.
                \STATE Set \(p=0\) and \(V=\emptyset\).
                \WHILE {\(\{1,2,\ldots,N\}\setminus V \neq \emptyset\)}
                    \IF {\(\abs{\{1,2,\ldots\}\setminus V}\geq M\)}
                        \STATE Set \(p=p+1\).
                        \STATE \text{Pick \(M\) elements from} \(\{1,2,\ldots,N\}\setminus V\) to create set \(v_{p}\).
                    \ELSE
                        \STATE Set \(p=p+1\).
                        \STATE Set \(\overline{v}_{p} = \{1,2,\ldots,N\}\subset V\).
                        \STATE Pick \(M-\abs{\overline{v}_{p}}\) elements from \(V\) to create \({v'}_{p}\).
                        \STATE Set \(v_{p} = \overline{v}_{p}\cup{v'}_{p}\).
                    \ENDIF
                    \STATE Set \(V = V\cup v_{p}\).
                \ENDWHILE

            \hspace*{-0.6cm}Step II: Construct a scheduling policy \(\gamma\) from \(v_{p}\), \(p=1,2,\ldots,r\) generated in Step I.
                \STATE Set \(b=0\) and \(\tau_{0} = 0\).
                \FOR {\(q=br,br+1,\ldots,(b+1)r-1\)}\label{step:repeat_step}
                    \STATE Set \(\gamma(\tau_{q}) = v_{q-br+1}\) and \(\tau_{q+1} = \tau_{q}+\alpha(\delta)\).
                    \STATE Output \(\tau_{q}\) and \(\gamma(\tau_{q})\).
                \ENDFOR
                \STATE Set \(b = b + 1\) and go to \ref{step:repeat_step}.			
		\end{algorithmic}
	\end{algorithm}

    Algorithm \ref{algo:scheduling} constructs a scheduling policy \(\gamma\) in two steps: In Step I, distinct vectors of \(M\) plants, \(v_{p}\), \(p=1,2,\ldots,r\) are created until each plant \(i\in\{1,2,\ldots,N\}\) appears at least once. Clearly, we obtain a total of \(r\) vectors, where \(r\) is as defined in \eqref{e:r_defn}. In Step II, a scheduling policy \(\gamma\) is constructed by employing \(v_{p}\), \(p=1,2,\ldots,r\). Each \(v_{p}\) is assigned to \(\gamma\) consecutively for \(\alpha(\delta)\) units of time. In view of Lemma \ref{lem:attack} this assignment ensures that the corresponding set of \(M\) plants operates in closed-loop for at least \(\delta\) units of time, where \(\delta\) is as described in Fact \ref{fact:key}.

    The following theorem provides sufficient conditions on the matrices \(A_{\is}\) and \(A_{\iu}\), \(i=1,2,\ldots,N\) under which a scheduling policy obtained from Algorithm \ref{algo:scheduling} preserves GES of each plant \(i\) in \eqref{e:plants}.
    \begin{thm}
    \label{t:mainres}
        Consider an NCS described in \S\ref{s:prob_stat}. Let the numbers \(M\), \(N\), \(k\) and \(m\) be given, and \(\lambda_{i}\), \(i=1,2,\ldots,N\) be arbitrary positive numbers satisfying
        \begin{align}
        \label{e:maincondn1}
            \rho e^{\lambda_{i}\delta} < 1.
        \end{align}
        Suppose that there exist scalars \(\varepsilon_{i}\), \(i=1,2,\ldots,N\) small enough such that the following set of conditions holds:
        \begin{align}
        \label{e:maincondn2}
            \norm{E_{\is\iu}}\leq\varepsilon_{i},
        \end{align}
        and
        \begin{align}
        \label{e:maincondn3}
            \rho e^{\lambda_{i}\delta} + \eta(\delta)\norm{A_{\is}}^{\delta-1}\norm{A_{\iu}}^{\zeta(\delta)}\varepsilon_{i}e^{\lambda_{i}r\alpha(\delta)}\leq 1.
        \end{align}
        Then a scheduling policy \(\gamma\) obtained from Algorithm \ref{algo:scheduling} preserves GES of each plant \(i\) in \eqref{e:plants}.
    \end{thm}

    \begin{proof}
        Fix a scheduling policy \(\gamma\) obtained from Algorithm \ref{algo:scheduling}. We will show that if conditions \eqref{e:maincondn1}-\eqref{e:maincondn3} hold, then \(\gamma\) preserves GES of each plant \(i\) in \eqref{e:plants} under all \(\kappa\in\A(m,k)\).

        Fix \(i\in\{1,2,\ldots,N\}\). The solution to \eqref{e:i-swsys} is given by
        \[
            x_{i}(t) = A_{\sigma_{i}(t-1)}\cdots A_{\sigma_{i}(1)}A_{\sigma_{i}(0)}x_{i}(0),\:\:t\in\N.
        \]
        Let \(W_{i}\) denote the matrix product on the right-hand side of the above equality. We let \(\overline{W}_{i}\) denote an initial segment of \(W_{i}\). Then condition \eqref{e:ges} can be writte equivalently as \cite[\S 2]{Agrachev2012}: for every initial segment \(\overline{W}_{i}\) of \(W_{i}\), we have
        \begin{align}
        \label{e:ges2}
            \norm{\overline{W}_{i}}\leq c_{i}e^{-\lambda_{i}\abs{\overline{W}_{i}}}\:\text{for all}\:\kappa\in\A(m,k).
        \end{align}
        It, therefore, suffices to show that condition \eqref{e:ges2} is true for \(\gamma\) under all admissible \(\kappa\). We will employ mathematical induction on \(\overline{W}_{i}\) to establish \eqref{e:ges2}.

        {\it A. Induction basis}: Pick \(c_{i}\) large enough so that \eqref{e:ges2} holds for all \(\overline{W}_{i}\) satisfying \(\overline{W}_{i}\leq r\alpha(\delta)\). Since there are only a finite number of such products corresponding to different \(\kappa\in\A(m,k)\), such a \(c_{i}\) can always be chosen.

        {\it B. Induction hypothesis}: Let \(\overline{W}_{i}\geq r\alpha(\delta)+1\) and assume that \eqref{e:ges2} is proved for all products of length less than \(\abs{\overline{W}_{i}}\).

        {\it C. Induction step}: Let \(W_{i} = L_{i}R_{i}\), where \(\abs{L_{i}} = r\alpha(\delta)\). We observe that \(L_{i}\) contains at least \(\delta\)-many \(A_{\is}\)'s. Indeed, by Lemma \ref{lem:attack}, it requires at most \(\alpha(\delta)\) time instants to ensure \(\delta\)-many closed-loop operations of the plant \(i\), and \(i\) appears in at least one \(v_{p}\), \(p\in\{1,2,\ldots,r\}\).

        We rewrite \(L_{i}\) as \(L_{i} = A_{\is}^{\delta}L_{i}^{(1)}+L_{i}^{(2)}\), where \(\abs{L_{i}^{(1)}} = r\alpha(\delta)-\delta\), and \(L_{i}^{(2)}\) contains at most \(\eta(\delta)\) terms of length \(r\alpha(\delta)-1\) each containing \(\zeta(\delta)\) \(A_{\iu}\)'s, \(\delta-1\) \(A_{\is}\)'s and \(1\) \(E_{\is\iu}\). Indeed, in the worst case, the following exchanges between the matrices \(A_{\is}\) and \(A_{\iu}\) are required to arrive at the desired structure of \(L_{i}\):
        \begin{itemize}[label = \(\circ\), leftmargin = *]
            \item \(m\)-many \(A_{\is}\)'s with \(k-m+(r-1)\alpha(\delta)\) \(A_{\iu}\)'s,
            \item \(m\)-many \(A_{\is}\)'s with \(2(k-m)+(r-1)\alpha(\delta)\) \(A_{\iu}\)'s,
            \item \(\cdots\),
            \item \(m\)-many \(A_{\is}\)'s with \(\lfloor\frac{\delta}{m}\rfloor(k-m)+(r-1)\alpha(\delta)\) \(A_{\iu}\)'s, and
            \item \(\biggl(\delta-\lfloor\frac{\delta}{m}\rfloor m\biggr)\)-many \(A_{\is}\)'s with \(\biggl(\lfloor\frac{\delta}{m}\rfloor+1\biggr)(k-m)+(r-1)\alpha(\delta)\) \(A_{\iu}\)'s.
        \end{itemize}

        Now, from the sub-multiplicativity and sub-additivity properties of the induced Euclidean norm, we have
        \begin{align}
            &\norm{\overline{W}_{i}}=\norm{L_{i}R_{i}}=\norm{\bigl(A_{\is}^{\delta}L_{i}^{(1)}+L_{i}^{(2)}\bigr)R_{i}}\nonumber\\
            &\leq\norm{A_{\is}^{\delta}}\norm{L_{i}^{(1)}R_{i}}+\norm{L_{i}^{(2)}}\norm{R_{i}}\nonumber\\
            &\leq\rho c_{i}e^{-\lambda_{i}(\abs{\overline{W}_{i}}-\delta)}+\eta(\delta)\norm{A_{\is}}^{\delta-1}\norm{A_{\iu}}^{\zeta(\delta)}
            \times\varepsilon_{i}c_{i}e^{-\lambda_{i}\bigl(\abs{\overline{W}_{i}}-(r-1)\alpha(\delta)\bigr)}\nonumber\\
            &=c_{i}e^{-\lambda_{i}\abs{\overline{W}_{i}}}\biggl(\rho e^{\lambda_{i}\delta}+\eta(\delta)\norm{A_{\is}}^{\delta-1}\norm{A_{\iu}}^{\zeta(\delta)}
            \times\varepsilon_{i}e^{\lambda_{i}r\alpha(\delta)}\biggr),\label{e:pf1_step2}
        \end{align}
        where the upper bounds on \(\norm{L_{i}^{(1)}R_{i}}\) and \(\norm{R_{i}}\) are obtained from the relations \(\abs{\overline{W}_{i}}=\abs{A_{\is}^{\delta}}+\abs{L_{i}^{(1)}R_{i}}\) and \(\abs{\overline{W}_{i}}=\abs{L_{i}}+\abs{R_{i}}\), respectively. Applying \eqref{e:maincondn3} to \eqref{e:pf1_step2} leads to \eqref{e:ges2}. Consequently, \eqref{e:i-swsys} is GES under \(\gamma\) for all \(\kappa\in\A(m,k)\).

        Recall that \(i\in\{1,2,\ldots,N\}\) was selected arbitrarily. It follows that the assertion of Theorem \ref{t:mainres} holds for all \(i\) in \eqref{e:plants}.
    \end{proof}

    For an NCS consisting of \(N\) discrete-time linear plants that are open-loop unstable and closed-loop stable and a shared communication network that allows access only to \(M\:(< N)\) plants at every time instant and is vulnerable to a maximum of \(m\) jamming attacks in every \(k\) consecutive time instants, Algorithm \ref{algo:scheduling} constructs a scheduling policy, and Theorem \ref{t:mainres} provides sufficient conditions for such a policy to be stabilizing. {S}ince \(\delta\geq m\) and \(\rho < 1\), there always exists a positive number \(\lambda_{i}\) (could be very small) such that \eqref{e:maincondn1} holds. For {a} plant \(i\) {to be stable under a scheduling policy \(\gamma\)}, we rely on the existence of a small enough scalar \(\varepsilon_{i}\) such that conditions \eqref{e:maincondn2}-\eqref{e:maincondn3} are satisfied. The scalar \(\varepsilon_{i}\) gives a measure of the ``closeness'' of the set of matrices \(\{A_{\is},A_{\iu}\}\) to a set of matrices \(\{A'_{\is},A'_{\iu}\}\) such that \(A'_{\is}\) and \(A'_{\iu}\) commute. In the simplest case when the matrices \(A_{\is}\) and \(A_{\iu}\) themselves commute (i.e., \(\varepsilon_{i} = 0\)), condition \eqref{e:maincondn3} reduces to condition \eqref{e:maincondn1}. Towards proving Theorem \ref{t:mainres}, we utilize the switched system model of the individual plants described in \eqref{e:i-swsys}, and show that a \(\gamma\) obtained from Algorithm \ref{algo:scheduling} is a concatenation of \(N\) stabilizing switching logics. Notice that our stability conditions do not ask for strict commutativity between the open-loop and closed-loop modes, \(A_{\iu}\) and \(A_{\is}\), of the plants, but rely on the ``closeness'' of these matrices to a set \(\{A'_{\is},A'_{\iu}\}\) such that \(A'_{\is}\) and \(A'_{\iu}\) commute, \(i\in\{1,2,\ldots,N\}\). This feature adds an inherent robustness to our stability conditions in the sense that if we are working with approximate models of \(A_{\is}\) and \(A_{\iu}\), or the elements of \(A_{\is}\) and \(A_{\iu}\) are prone to change over time, then GES of plant \(i\) is preserved under our scheduling policies as long as the Euclidean norm of \(E_{\is\iu}\) is small enough to satisfy condition \eqref{e:maincondn3}.

    \begin{remark}
    \label{rem:swsys}
    \rm{
        Switched systems have appeared before in the scheduling literature, see e.g., \cite{Hristu2001,Lin2005,abc} and the references therein. In particular, the class of average dwell time switching logics is proven to be a useful tool in the design of stabilizing scheduling policies for NCSs with continuous-time plants, see e.g., \cite{Lin2005}. In the discrete-time setting, recently in \cite{abc} we employed switched systems and graph theory to design stabilizing periodic scheduling policies for NCSs under ideal communication between the plants and their controllers. The design of such policies involves what is called \(T\)-contractive cycles on the underlying weighted directed graph of an NCS. In contrast, in this paper we accommodate jamming attacks to the shared network in our design of periodic scheduling policies, and rely on commutation relations between the stable and unstable dynamics of each plant to guarantee stability.
        }
    \end{remark}

     \begin{remark}
    \label{rem:swsys}
    \rm{
        Matrix commutators (Lie brackets) have been employed in the context of stability of switched systems earlier in the literature, see e.g., \cite{Narendra1994, Agrachev2012, xyz1, xyz2, ghi}. In this paper we employ this tool towards achieving simultaneous stability of \(N\) switched systems. The ``worst case'' behaviour of \(\kappa\) (i.e., in every \(k\) consecutive time instants, some or all of the control inputs are deactivated exactly at \(m\) time instants) and \(\gamma\) (i.e., \(i\) appears exactly in one \(v_{p}\)) impose an upper bound on the number of consecutive time instants with \(\sigma_{i}(t) = \iu\) between every two instances of \(\sigma_{i}(t') = \sigma_{i}(t'') = \is\), satisfying \(\sigma_{i}(\tau) \neq \is\) for all \(\tau\in[t',t'']\), \(i=1,2,\ldots,N\). For each \(i\), we consider two subsystems --- one stable and one unstable, and employ the above bound in our design of stabilizing switching logics \(\sigma_{i}\), \(i=1,2,\ldots,N\). The design of stabilizing switching logics with constrained activation of unstable subsystems by employing (matrix) commutation relations between the subsystems matrices is addressed earlier in \cite{xyz2,ghi}. The stability conditions of \cite{xyz2,ghi} and the current paper are only sufficient, and differ because of the type of counting argument employed in their derivations.
    }
    \end{remark}

    \begin{remark}
    \label{rem:plant_selection}
    \rm{
        Notice that in Step I of Algorithm \ref{algo:scheduling}, {it holds that} if \(\displaystyle{N-\lfloor{N}/{M}\rfloor M}\)\(\displaystyle{\neq 0}\), then for \(\displaystyle{M-(N-\lfloor{N}/{M}\rfloor M})\)-many \(i\in\{1,2,\ldots,N\}\) there exist \(a_{i},b_{i}\in\{1,2,\ldots,r\}\) such that \(i\) {is} an element of both \(v_{a_{i}}\) and \(v_{b_{i}}\). Indeed, by construction, the vectors \(\displaystyle{v_{1},v_{2},\ldots,v_{\lfloor{N}/{M}\rfloor}}\) are assigned \(\displaystyle{\lfloor{N}/{M}\rfloor M}\)-many distinct elements, and the vector \(\displaystyle{v_{\lfloor{N}/{M}\rfloor+1}}\) is assigned the remaining \(\displaystyle{N-\lfloor{N}/{M}\rfloor M}\)-many elements and\\ \(\displaystyle{M-(N-}\)\({\lfloor{N}/{M}\rfloor M)}\)-many elements from the ones that appeared in \(v_{p}\), \(\displaystyle{p=1,2,\ldots}\),\\\({\lfloor{N}/{M}\rfloor}\). There is, therefore, an element of ``choice'' associated to the selection of the above \(M-(N-\lfloor{N}/{M}\rfloor M)\)-many elements. In this paper we are concerned with stability of each plant \(i\) in \eqref{e:plants}, and consequently, any valid selection serves our purpose. However, in general, the selection of \(\displaystyle{M-(N-\lfloor{N}/{M}\rfloor M)}\)-many elements from the ones that appeared in \(v_{p}\), \(\displaystyle{p=1,2,\ldots,\lfloor{N}/{M}\rfloor}\) can be governed by the type of qualitative or quantitative property of each plant that one wishes to be preserved by \(\gamma\). For instance, if a plant \(j\) is more unstable than the other plants \(\{1,2,\ldots,N\}\setminus\{j\}\), then one may want to include \(j\) in the \(\displaystyle{M-(N-\lfloor{N}/{M}\rfloor M)}\)-many elements described above to ensure more stable operations of the plant \(j\).
    }
    \end{remark}
\section{Numerical experiments}
\label{s:numex}
    \begin{example}
    \label{ex:numex1}
    \rm{
    Let \(N=5\). We generate the matrices \(A_{i}\in\R^{2\times 2}\), \(B_{i}\in\R^{2\times 1}\) and \(K_{i}\in\R^{1\times 2}\), \(i=1,2,3,4,5\) as follows:
    \begin{itemize}[label = \(\circ\), leftmargin = *]
        \item The elements of \(A_{i}\) are selected from the interval \([-2,2]\) uniformly at random.
        \item Elements of \(B_{i}\) are selected by picking values from the set \(\{0,1\}\).
        \item It is ensured that the pair \((A_{i},B_{i})\) is controllable, and \(K_{i}\) is the discrete-time quadratic regulator for \((A_{i},B_{i})\) with \(Q_{i} = Q = 5I_{2\times 2}\) and \(R_{i} = R = 1\).
    \end{itemize}
    The corresponding numerical values are given in Table \ref{tab:plant_data}. Let \(M = 2\), \(k=2\) and \(m=1\).
    \begin{table*}[htbp]
	\centering
	\begin{tabular}{|c | c | c|c|c|c|}
		\hline
		\(i\) & \(A_{i}\) & \(B_{i}\) & \(K_{i}\) & \(\abs{\lambda(A_{i})}\) & \(\abs{\lambda(A_{i}+B_{i}K_{i})}\)\\
		\hline
		\(1\) & \(\pmat{-0.1340 & -0.0076\\-0.0503 &  -1.0821}\) & \(\pmat{0\\1}\) & \(\pmat{0.0441 & 0.9277}\) & \(0.1336,1.0825\) & \(0.1319,0.1565\)\\
		\hline
		\(2\) & \(\pmat{-0.0107 & 0.0052\\0.4219 &  1.0993}\) & \(\pmat{0\\1}\) & \(\pmat{-0.3613 & -0.9434}\) & \(0.0127,1.1013\) & \(0.0126,0.1578\)\\
		\hline
		\(3\) & \(\pmat{0.6505 & 0.4401\\0.6510 &  0.4197}\) & \(\pmat{1\\1}\) & \(\pmat{-0.5968 & -0.3945}\) & \(1.0827,0.0125\) & \(0.0912,0.0123\)\\
		\hline
		\(4\) & \(\pmat{-1.3188 & -0.1959\\0.0008 &  -0.0244}\) & \(\pmat{1\\0}\) & \(\pmat{1.1431 & 0.1705}\) & \(1.3187,0.0245\) & \(0.1756,0.0245\)\\
		\hline
		\(5\) & \(\pmat{-1.0079 & -0.0455\\0.0266 &  0.6821}\) & \(\pmat{1\\0}\) & \(\pmat{0.8605 & 0.0201}\) & \(1.0072,0.6814\) & \(0.1466,0.6813\)\\
		\hline
	\end{tabular}
    \vspace*{0.2cm}
	\caption{Description of individual plants in the NCS}\label{tab:plant_data}
	\end{table*}
    We have \(\norm{A_{1_{s}}} = 0.1565\), \(\norm{A_{1_{u}}} = 1.0834\), \(\norm{E_{1_{s}1_{u}}} = 0.0071\), \(\norm{A_{2_{s}}} = 0.1673\), \(\norm{A_{2_{u}}} = 1.1775\), \(\norm{E_{2_{s}2_{u}}} = 0.0053\), \(\norm{A_{3_{s}}} = 0.0916\), \(\norm{A_{3_{u}}} = 1.1030\), \(\norm{E_{3_{s}3_{u}}} = 0.0102\), \(\norm{A_{4_{s}}} = 0.1775\), \(\norm{A_{4_{u}}} = 1.3333\), \(\norm{E_{4_{s}4_{u}}} = 0.0032\), \(\norm{A_{5_{s}}} = 0.6834\), \(\norm{A_{5_{u}}} = 1.0119\), \(\norm{E_{5_{s}5_{u}}} = 0.0209\). Also, \(r=3\), \(\delta = 1\), \(\alpha(\delta) = 3\), \(\zeta(\delta) = 7\), and \(\eta(\delta) = 7\).

    We pick \(\rho = 0.7\) and \(\lambda_{i} = 0.0001\), \(i=1,2,3,4,5\), which leads to
    \[
        \rho e^{\lambda_{i}\delta} = 0.0370,\:\:i=1,2,3,4,5.
    \]
    Let \(\varepsilon_{1} = 0.008\), \(\varepsilon_{2} = 0.006\), \(\varepsilon_{3} = 0.02\), \(\varepsilon_{4} = 0.004\), \(\varepsilon_{5} = 0.03\).

    From the above set of numerical values, we obtain
    \begin{align*}
        \rho e^{\lambda_{1}\delta} + \eta(\delta)\norm{A_{1_{s}}}^{\delta-1}\norm{A_{1_{u}}}^{\zeta(\delta)}\varepsilon_{1}e^{\lambda_{1}r\alpha(\delta)} &= 0.7981,\\
         \rho e^{\lambda_{2}\delta} + \eta(\delta)\norm{A_{2_{s}}}^{\delta-1}\norm{A_{2_{u}}}^{\zeta(\delta)}\varepsilon_{2}e^{\lambda_{2}r\alpha(\delta)} &= 0.8318,\\
          \rho e^{\lambda_{3}\delta} + \eta(\delta)\norm{A_{3_{s}}}^{\delta-1}\norm{A_{3_{u}}}^{\zeta(\delta)}\varepsilon_{3}e^{\lambda_{3}r\alpha(\delta)} &= 0.9781,\\
           \rho e^{\lambda_{4}\delta} + \eta(\delta)\norm{A_{4_{s}}}^{\delta-1}\norm{A_{4_{u}}}^{\zeta(\delta)}\varepsilon_{4}e^{\lambda_{4}r\alpha(\delta)} &= 0.9098,\\
            \rho e^{\lambda_{5}\delta} + \eta(\delta)\norm{A_{5_{s}}}^{\delta-1}\norm{A_{5_{u}}}^{\zeta(\delta)}\varepsilon_{5}e^{\lambda_{5}r\alpha(\delta)} &= 0.9282.
    \end{align*}
    Consequently, conditions \eqref{e:maincondn1}-\eqref{e:maincondn3} hold, and the assertion of Theorem \ref{t:mainres} follows. We now employ Algorithm \ref{algo:scheduling} to design scheduling policies for the NCS under consideration.\\
    Step I: We obtain \(v_{1} = \{1,3\}\), \(v_{2} = \{2,5\}\) and \(v_{3} = \{1,4\}\).\\
    Step II: A scheduling policy \(\gamma\) is designed with the above choice of \(v_{p}\), \(p=1,2,3\). Then \(10\) different attack signals \(\kappa\in\A(m,k)\) are generated randomly. Corresponding to each \(\kappa\) above, we plot \((\norm{x_{i}(t)})_{t\in\N_{0}}\) for \(10\) different initial conditions \(x_{i}(0)\) chosen uniformly at random from the interval \([-1,1]^{2}\), \(i=1,2,\ldots,N\). The plots are illustrated in Figures \ref{fig:plant1}-\ref{fig:plant5}.
    \begin{figure}
    \centering
        \includegraphics[scale = 0.35]{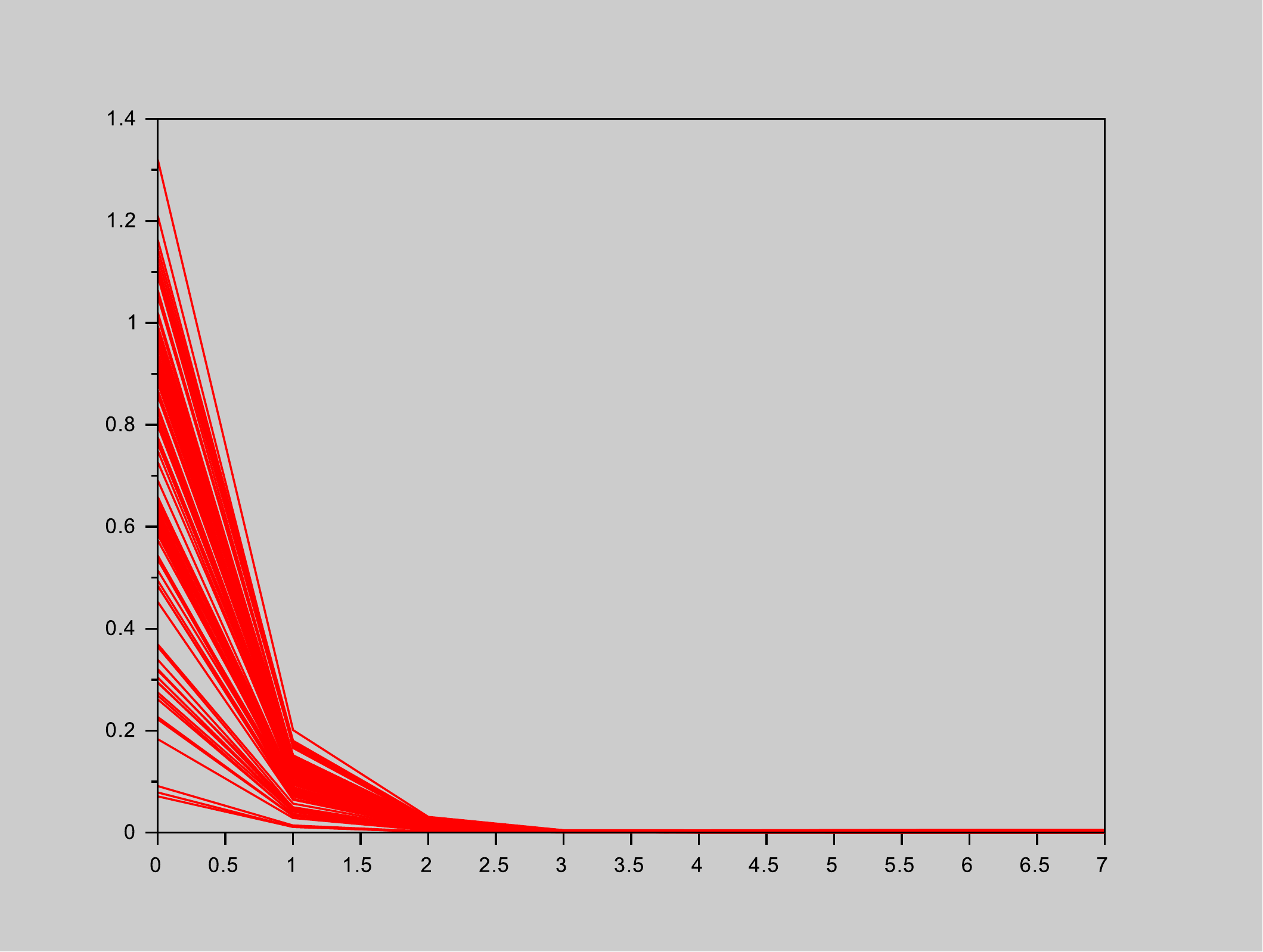}
        \caption{Plot of \((\norm{x_{1}(t)})_{t\in\N_{0}}\)}\label{fig:plant1}
    \end{figure}
    \begin{figure}
    \centering
        \includegraphics[scale = 0.35]{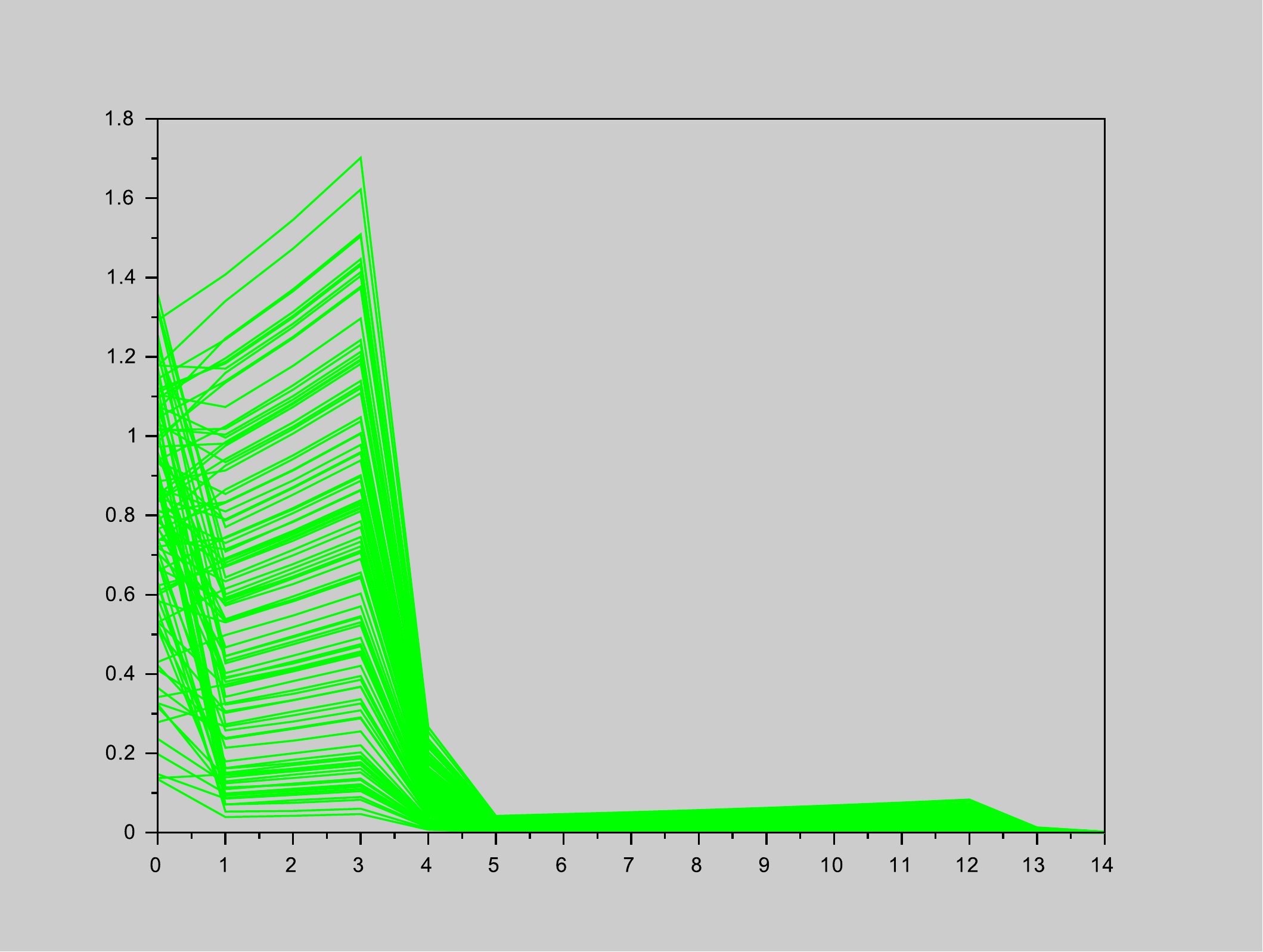}
        \caption{Plot of \((\norm{x_{2}(t)})_{t\in\N_{0}}\)}\label{fig:plant2}
    \end{figure}
    \begin{figure}
    \centering
        \includegraphics[scale = 0.35]{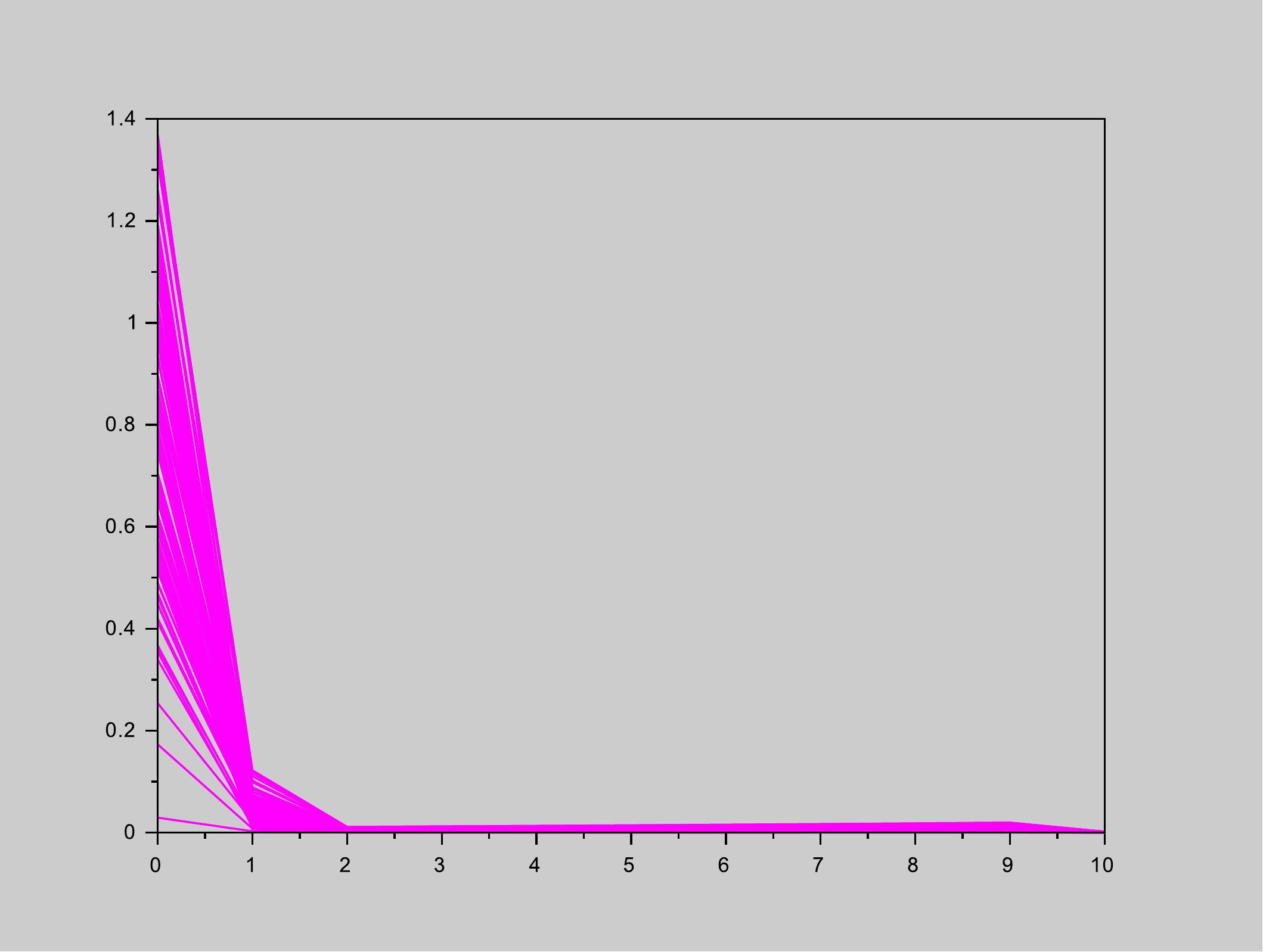}
        \caption{Plot of \((\norm{x_{3}(t)})_{t\in\N_{0}}\)}\label{fig:plant3}
    \end{figure}
    \begin{figure}
    \centering
        \includegraphics[scale = 0.35]{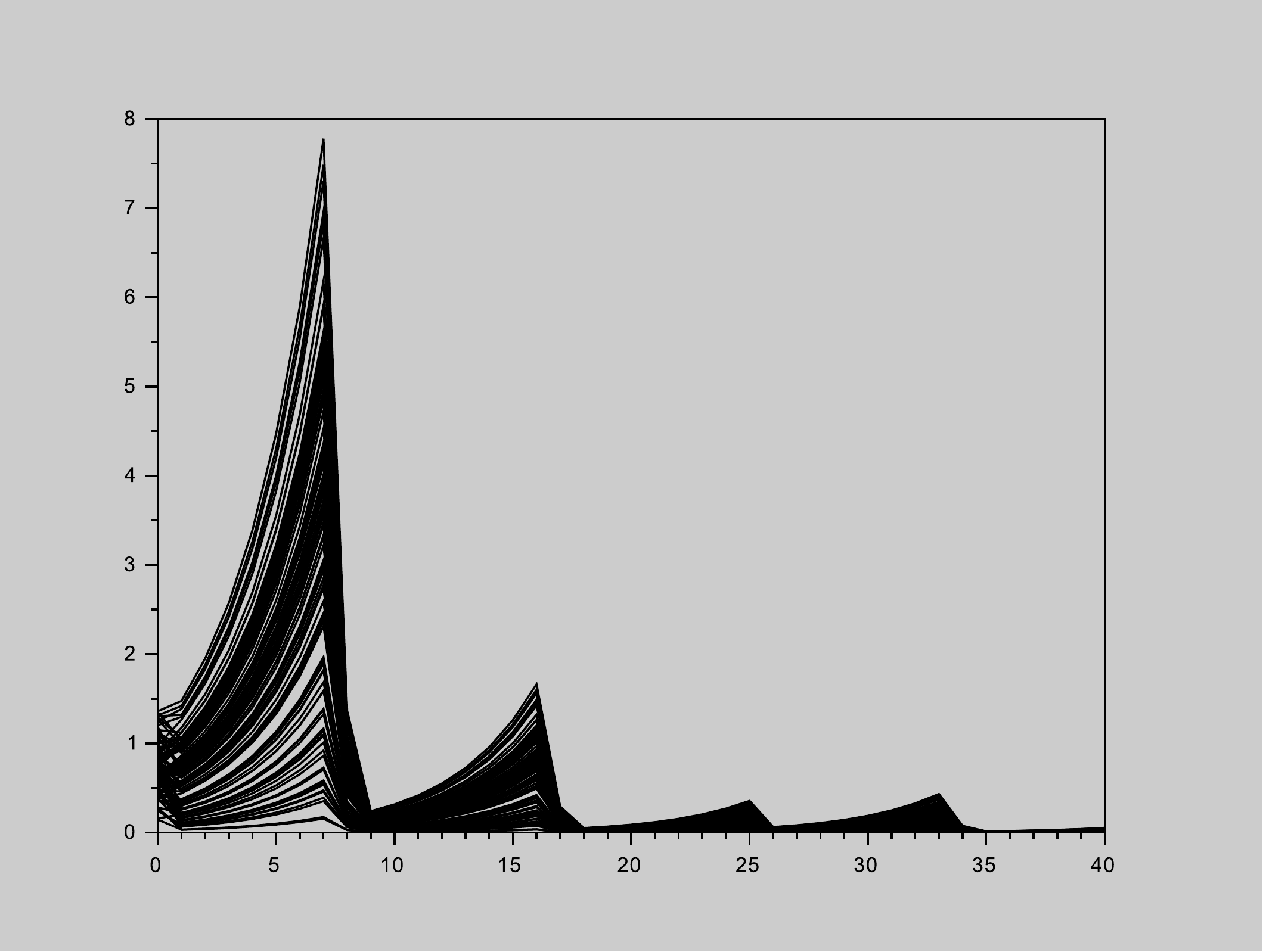}
        \caption{Plot of \((\norm{x_{4}(t)})_{t\in\N_{0}}\)}\label{fig:plant4}
    \end{figure}
    \begin{figure}
    \centering
        \includegraphics[scale = 0.35]{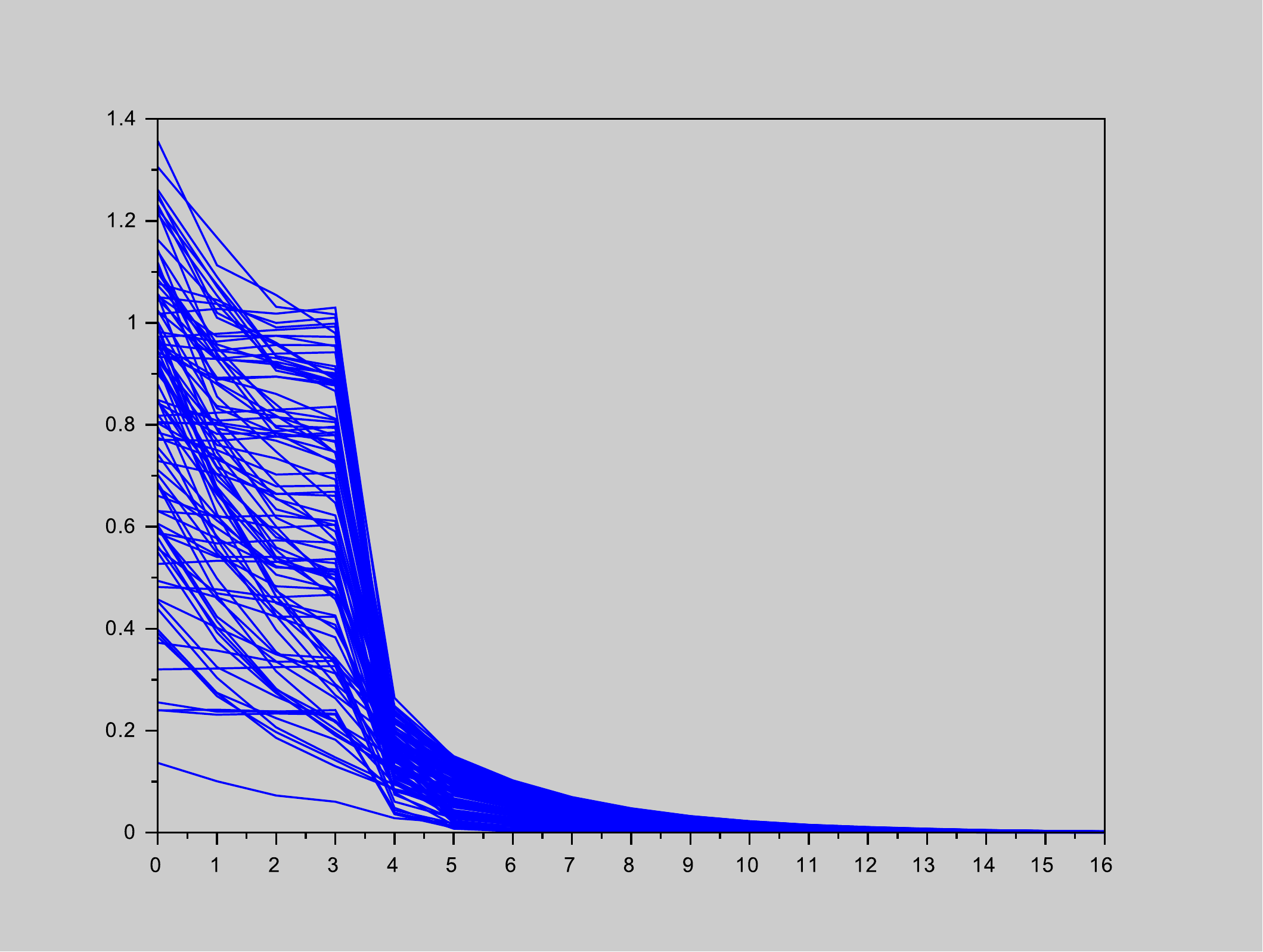}
        \caption{Plot of \((\norm{x_{5}(t)})_{t\in\N_{0}}\)}\label{fig:plant5}
    \end{figure}
    }
    \end{example}

    \begin{example}
    \label{ex:numex2}
    \rm{
        We now perform a set of experiments to demonstrate the scalability of our stability conditions. Notice that if the matrices \(A_{\is}\) and \(A_{\iu}\) commute (i.e., \(\varepsilon_{i} = 0\)), \(i\in\{1,2,\ldots,N\}\), then condition \eqref{e:maincondn3} holds trivially. We are, however, interested in the non-commuting case, i.e., estimates of upper bounds on the scalar \(\varepsilon_{i}\), \(i\in\{1,2,\ldots,N\}\) for large-scale settings. We study the following two cases:
        \begin{enumerate}[label = \Roman*., leftmargin = *]
            \item We fix \(\norm{A_{\is}} = 0.9\), \(\norm{A_{\iu}} = 1.1\), \(\rho = 0.9\), \(\lambda_{i} = 0.00001\), \(k=3\), \(m=2\), and vary \(\delta\) and \(r\) to compute maximum values of \(\varepsilon_{i}\) under which condition \eqref{e:maincondn3} holds. The corresponding plot is given in Figure \ref{fig:ls-plot1}.
            \item We fix \(r=2\), \(k=5\), \(\norm{A_{\is}} = 0.9\), \(\norm{A_{\iu}} = 1.1\), \(\rho = 0.9\), \(\lambda_{i} = 0.00001\), and vary \(\delta\) and \(m\) \((m < k)\) to compute maximum values of \(\varepsilon_{i}\) for the satisfaction of \eqref{e:maincondn3}. The corresponding plot is given in Figure \ref{fig:ls-plot2}.
        \end{enumerate}
        Not surprisingly, it is observed that the size of the set of non-commuting pairs of \(A_{\is}\) and \(A_{\iu}\) that Algorithm \ref{algo:scheduling} caters to, shrinks with an increase in \(r\) (i.e., with how large \(N-M\) is), and the attack signal parameters, \(k\) and \(m\).
    }
    \end{example}
    \begin{figure}
    \centering
        \includegraphics[scale = 0.35]{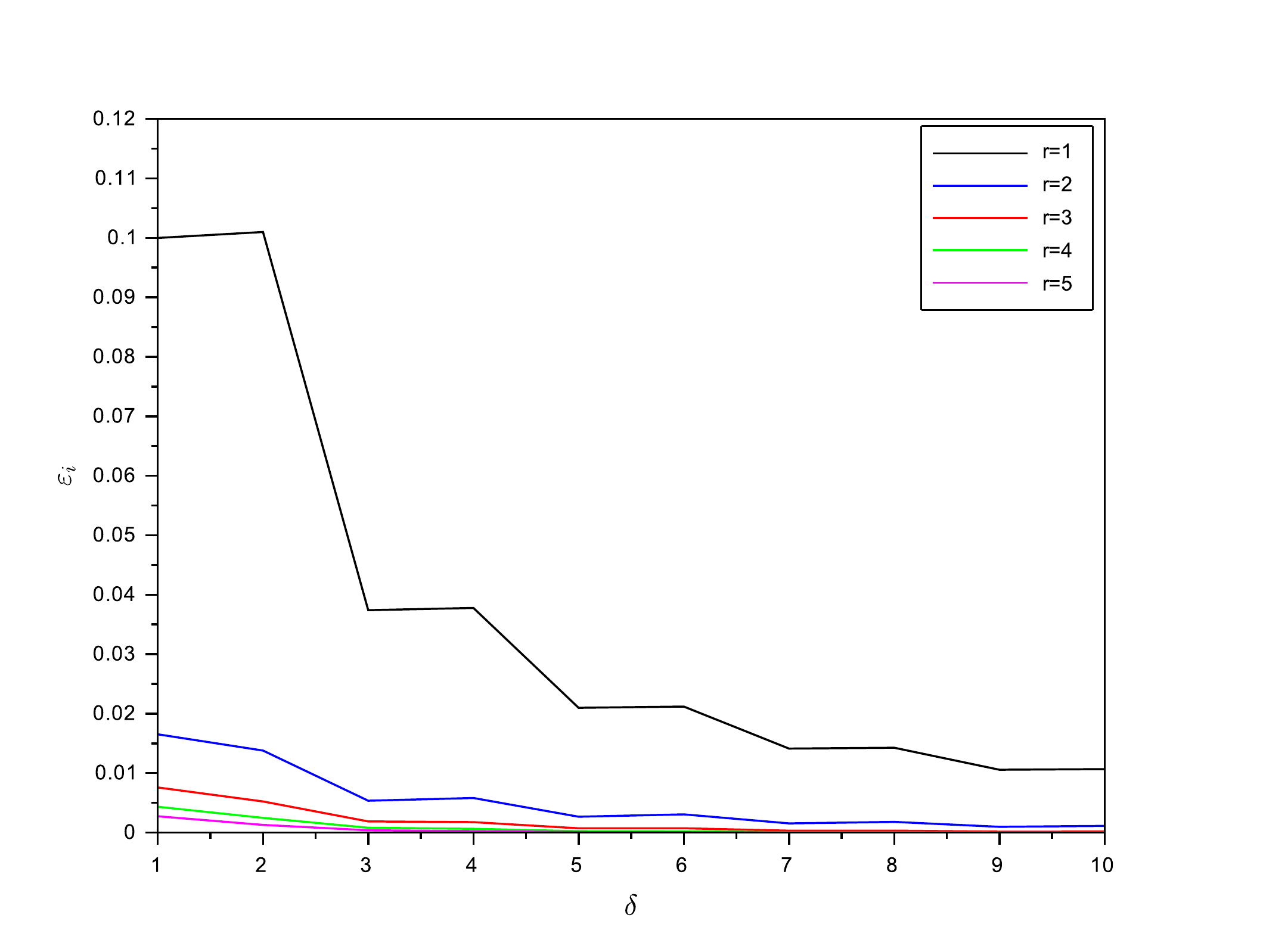}
        \caption{Plot of maximum values of \(\varepsilon_{i}\) for the satisfaction of \eqref{e:maincondn3}}\label{fig:ls-plot1}
    \end{figure}
    \begin{figure}
    \centering
        \includegraphics[scale = 0.35]{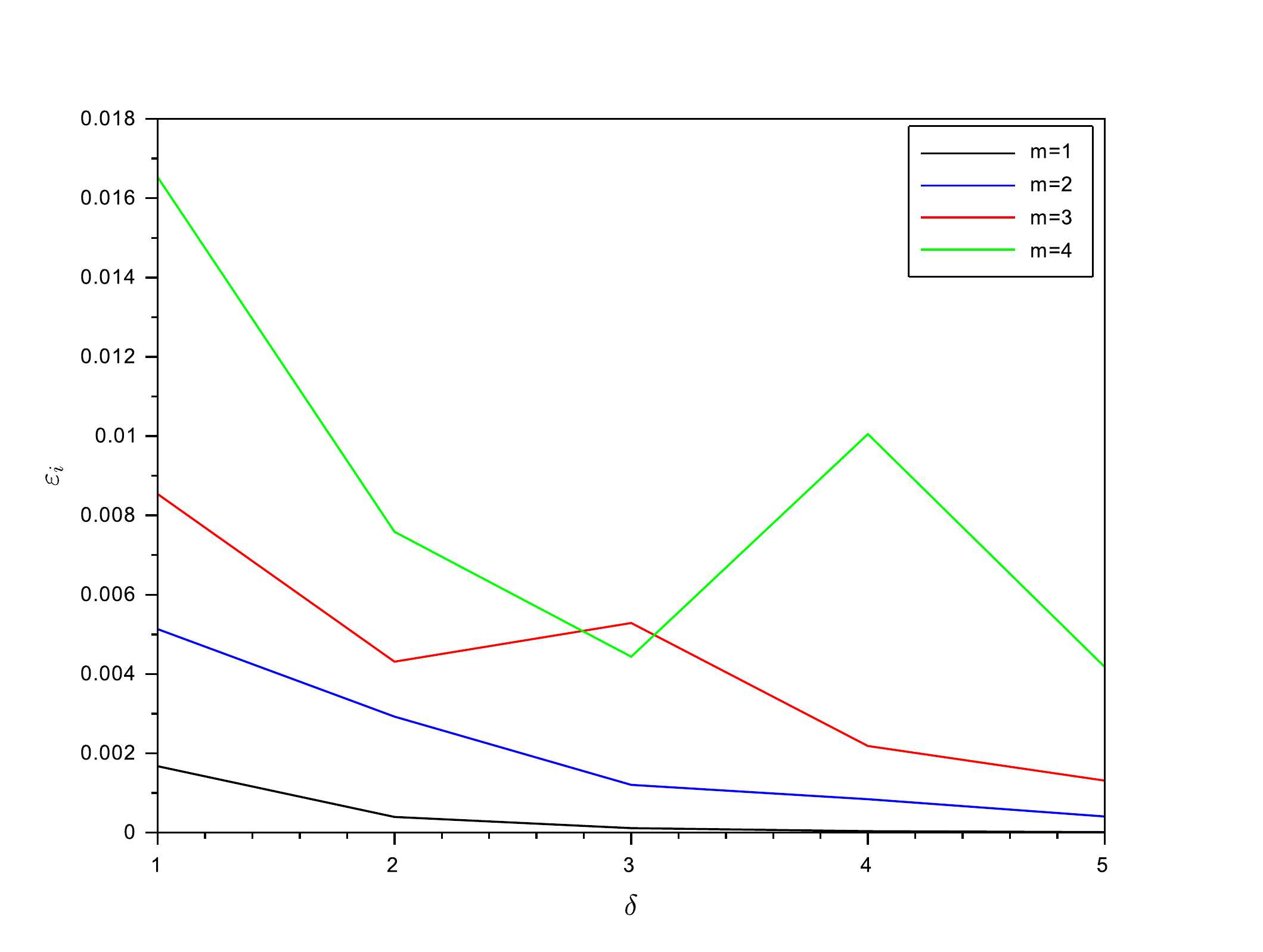}
        \caption{Plot of maximum values of \(\varepsilon_{i}\) for the satisfaction of \eqref{e:maincondn3}}\label{fig:ls-plot2}
    \end{figure}

\section{Conclusion}
\label{s:concln}	
    In this paper we designed stabilizing scheduling policies for NCSs whose communication networks have limited bandwidth and are vulnerable to jamming attacks. A next natural question is regarding the controller design: given the pairs of matrices \((A_{i},B_{i})\), \(i=1,2,\ldots,N\), and the numbers \(M\), \(m\) and \(k\), how do we choose \(K_{i}\), \(i=1,2,\ldots,N\), such that each plant \(i\) in \eqref{e:plants} satisfies conditions \eqref{e:maincondn2}-\eqref{e:maincondn3}? This problem is currently under investigation and the findings will be reported elsewhere.



\end{document}